%% file: main.tex
\documentclass[11pt]{article}

\usepackage{amsmath}
\usepackage{unit}
\usepackage{amssymb}
\usepackage{mathrsfs}
\usepackage[dvipdfmx]{graphicx}
\usepackage{setspace}
\usepackage{verbatim}
\usepackage{url}
\usepackage{dcolumn}
\usepackage{color}
\usepackage{here}
\usepackage{subcaption}
\usepackage{amsthm}
\usepackage{bm}
\usepackage{amscd}
\theoremstyle{definition}

\newtheorem{remark}{Remark}

\newtheorem{theorem}{\bf Theorem}[section]

\graphicspath{{./Figure/}}

\title{%
Asymptotic assessment of distribution voltage profile\\
using a nonlinear ODE model\thanks{Y.S. acknowledges supports from Japan Science and Technology Agency, Core Research for Evolutional Science and Technology (JST-CREST) Program \#JP-MJCR15K3, and from Japan Society for the Promotion of Science (JSPS), KAKENHI Grant \#20K04552.}
}%

\author{
Haruki Tadano, 
Yoshihiko Susuki\footnote{Corresponding author, {\tt susuki@eis.osakafu-u.ac.jp}}, 
and Atsushi Ishigame\footnote{They are with Department of Electrical and Information Systems, Osaka Prefecture University, 1-1 Gakuen-cho, Naka-ku, Sakai, 599-8631 Japan.}
}%

\date{}

\begin{document}
\maketitle

\begin{abstract}
   The promising increase of Electric Vehicles (EVs) in our society poses a challenging problem on the impact assessment of their charging/discharging to power distribution grids.
   This paper addresses the assessment problem in a framework of nonlinear differential equations.
   Specifically, we address the nonlinear ODE (Ordinary Differential Equation) model for representing the spatial profile of voltage phasor along a distribution feeder, which has been recently introduced in literature.
   The assessment problem is then formulated as a two-point boundary value problem of the nonlinear ODE model.
   In this paper we then derive an asymptotic charcterisation of solutions of the problem through the standard regular perturbation method.
   This provides a mathematically-rigor and quantitative method for assessing how the charging/discharging of EVs affects the spatial profile of distribution voltage.
   Effectiveness of the asymptotic charcterisation is established with simulations of both simple and practical configurations of the power distribution grid.
\end{abstract}


\section{Introduction}
\label{sec:intro}

\emph{Technological motivation.---}The high penetration of Electric Vehicles (EVs) is a promising future in our society \cite{Net-Zero}.
It is reported in \cite{GlobalOutlook} that the power demand from global EV fleet reached the total electricity consumption in Germany and the Netherlands in 2017.
In this, concerns with their substantial impacts to power distribution grids have raised such as congestion management and voltage amplitude regulation: see, e.g., \cite{Ipakchi,Arias:2019}.
In the Nordic region where EVs are penetrating in advance, it is pointed out in \cite{NordicOutlook} that there is a possibility of overloading distribution transformers in urban areas due to EV charging in the future.

The concerns pose a problem on the impact assessment of EV charging to power distribution grids.
There are multiple reasons why it is now important and challenging.
A new technology called \emph{fast charging} has been developed with the high increase of the amount of charged power with grid-faced inverters \cite{GlobalOutlook}.
Also, the so-called ancillary service using a cooperative use of a large population of EVs has been developed to provide the fast responsiveness of frequency control in power transmission grids: see, e.g., \cite{Kempton,Tomic,Ota}.
This kind of technology related to vehicles is generally referred to as V2X \cite{Toh}.
An EV is regarded as an autonomously moving battery in the spatial domain and can conduct the charging and discharging (in principle) \emph{anywhere} in a distribution grid, which has a large number of spatially distributed connecting points such as households, charging points in shopping malls, and charging stations.
This is completely different from other Distributed Energy Resources (DERs) such as Photo-Voltaic (PV) generation units, which do not move spatially.
The impact assessment is thus important to solve from the new power technologies and challenging as a new problem arising in a mixed domain of power and transportation systems.

The impact, in particular, to the voltage amplitude of power distribution feeders, is historically assessed with the so-called power-flow equation that is an algebraic mathematical model for the discretized evaluation of distribution voltage: see, e.g., \cite{Kersting}.
Although the model is nonlinear, it has been widely used for the impact assessment due to EV charging: see \cite{Clement1,Clement2,Yilmaz:2013,Arias:2019,Dixon:2020} and references therein.
However, the assessment is numerics-based, still computationally costly, and does not provide information of the impact with a clear reference to its physical origin.
In particular, it is hard for us to gain quantitative measures on the spatial impact on distribution voltage, such as how far EV charging at a particular location affects the distribution voltage, which is crucial to the current assessment problem regarding the autonomously moving battery.
A new modeling of the spatial profile of distribution voltage---distribution voltage profile---and associated assessment methodology are therefore required.

\emph{Purpose and contributions.---}The purpose of this paper is to solve the assessment problem in a framework of nonlinear differential equations.
Specifically, we address the nonlinear Ordinary Differential Equation (ODE) model for representing the distribution voltage profile derived by Chertkov \emph{et al.}  \cite{Chertkov}.
Unlike the power-flow equation, the nonlinear ODE model is capable of representing the intrinsic spatial (continuous in space) characteristics of distribution voltage profile.
The nonlinear ODE model therefore explicitly keeps spatial information of (balanced) distribution grids and hence enables us to quantify the spatial impact of EVs on the distribution voltage profile.
The nonlinear ODE is used for evaluating and mitigating the impact of DERs including PV units and EVs \cite{Baek,Susuki1,Susuki2,Mizuta2}.
In \cite{Chertkov} the authors formulate the assessment problem as a \emph{Two-Point Boundary Value} (TPBV) problem of the nonlinear ODE, and in \cite{Baek,Susuki1} the authors provide a numerical scheme for approximately deriving its solution via discretisation.
The boundary value problem of nonlinear ODEs has
been historically studied in applied and computational mathematics \cite{Keller,Domokos,Chowdhury}.
In this paper, as the theoretical foundation of the preceding work \cite{Baek,Susuki1,Susuki2,Mizuta2}, we characterize the solution of the nonlinear TPBV problem using the regular perturbation technique \cite{Guckenheimer} and derive a sequence of Initial Value (IV) problems of linear ODEs whose solutions asymptotically approximate the original solution.
This has benefits from the technological viewpoint: for instance, as shown in this paper, it enables us to quantify the impact of EV charging in a separation manner from the those of loads and others DERs.
This is never archived with the power-flow equations and one of the novelty of our ODE approach.
Effectiveness of the asymptotic charcterisation is established with simulations of both simple and practical configurations of the distribution grid.

The contributions of the paper are three-fold.
First, we newly derive an asymptotic representation of the distribution voltage profile by applying the regular perturbation technique to the nonlinear ODE model.
It enables us to approximately evaluate the profile, which is the solution of the nonlinear TPBV problem, by solving solutions of the associated IV problems of linear ODEs.
Needless to say, the linear ODEs are simple to solve both analytically and numerically, and are thus direct to the assessment in complex power grids.
Second, to provide the mathematical background of the representation, we collect a series of proofs for existence of solutions for the nonlinear TPBV problem and associated IV problems of linear ODEs.
A regularity result for the solutions is also proved.
Third, we demonstrate effectiveness of the representation with numerical simulations of simple and practical configurations of the distribution grid.
Preliminary work of this paper is presented in \cite{Tadano} as a non-reviewed report in a domestic conference.
This paper is a substantially enhanced version of \cite{Tadano} by newly adding a series of proofs for the existence and regularity of solutions and presenting a new set of simulation results for the practical power grid model.

{\emph{Organisation of this paper.---}Section~\ref{sec:ODE} introduces the nonlinear ODE model for distribution voltage profile and states the TPBV problem of the nonlinear ODE.
In Section~\ref{sec:perturb}, we derive the asymptotic representation of solutions of the nonlinear problem using the standard perturbation technique.
A series of theoretical results on solutions of the original nonlinear TPBV and associated linear IV problems are also derived.
Section~\ref{sec:numerics} presents numerical simulations to validate the asymptotic expansion result in Section~\ref{sec:perturb}.
The conclusion is made in Section~\ref{sec:outro} with a brief summery and future directions.
%

\section{ODE Model of Distribution Voltage Profile}
\label{sec:ODE}
\begin{figure}[t]
\centering
\includegraphics[width=0.7\hsize]{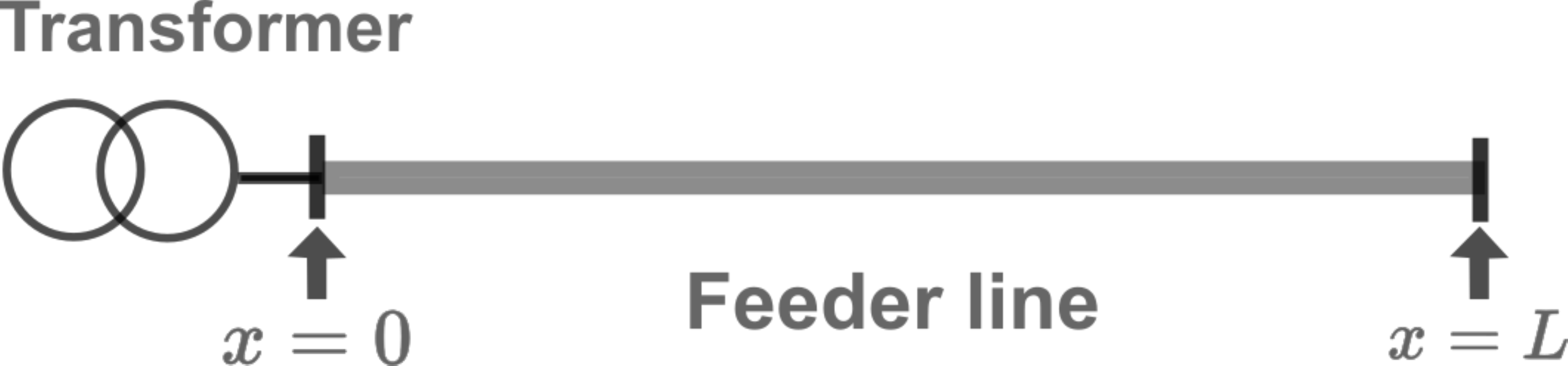}
\caption{Balanced, straight-line distribution feeder that starts a substation transformer (bank) and ends at a non-loading point.}
\label{fig:simplefeeder}
\end{figure}

At the beginning, we introduce a mathematical approach to model the voltage profile of distribution systems 
in \cite{Chertkov}.
In this paper, our concerns pose a problem on the impact assessment of the spatially distributed EV charging at each location to distribution systems.
Thus, it becomes relevant to consider the voltage profile starting at a distribution substation (bank) that is continuous in space (length).
Now consider a single distribution feeder shown in Figure~\ref{fig:simplefeeder}, starting at a transformer where 
the origin of the one-dimensional displacement (location) is introroduced.
In AC electrical networks, 
phasor representations of voltage amplitide and phase are used as physical quantities. 
The voltage phasor at the location $x$ is represented with $v(x)\ee^{\ii\theta(x)}$, where $\ii$ stands for the imaginary unit, $v(x)$ the voltage amplitude [V], and $\theta(x)$ the voltage phase [rad].
At the starting point $x=0$ in Figure~\ref{fig:simplefeeder}, due to voltage regulation at the substation, we naturally set $v(0)$ to be constant.
Throughout this paper, $v(0)$ coincides with unity in per-unit system and $\theta(0)$ with zero as a reference.
Then, the two functions $\theta$ and $v$ are described as the following nonlinear ODE \cite{Chertkov}:
\begin{subequations}
\label{eqn:ode_tv}
\begin{align}
-\frac{\dd}{\dd x}\left(v^2\frac{\dd\theta}{\dd x}\right) &= \frac{b(x)p(x)-g(x)q(x)}{g(x)^2+b(x)^2},
\\
\frac{\dd^2 v}{\dd x^2} &= v\left(\frac{\dd\theta}{\dd x}\right)^2
-\frac{g(x)p(x)+b(x)q(x)}{v(g(x)^2+b(x)^2)}.
\end{align}
\end{subequations}
The functions $g(x)$ and $b(x)$ in \eqref{eqn:ode_sw} are the position-dependent conductance and susceptance per unit-length [S/km] and assumed to be continuous in $x$.
Also, the functions $p(x)$ (or $q(x)$) is the active (or reactive) power flowing into the feeder (note that $p(x)>$ 0 indicates the positive active power flowing to the feeder at $x$).
In this paper, we refer to $p(x)$ and $q(x)$ as the power density functions in [W/km] and [Var/km].
Here, because Eq.~\eqref{eqn:ode_tv} is complicated, the two ancillary functions $s(x)$ and $w(x)$ are introduced as
\begin{subequations}
\label{eqn:ode_sw}
\begin{align}
s(x) &:= -v(x)^2\frac{\dd}{\dd x}\theta(x), 
\\
w(x) &:= \frac{\dd}{\dd x}v(x). \label{eqn:dw/dx}
\end{align}
\end{subequations}
The function $w(x)$ is called the voltage gradient [V/km].
At the end point $x=L$ in Figure~\ref{fig:simplefeeder} (namely, no feeder and load exist for $x>L$), by supposing that the end is not loaded, we have the conditions $s(L)=0$ and $w(L)=0$.
%
With the above, the nonlinear ODE model of distribution voltage profile is derived as
\begin{subequations}
\label{eqn:ode}
\begin{align}
\frac{\dd\theta}{\dd x} &= -\frac{s}{v^2}~~~(v\neq 0) \label{eqn:dt/dx},
\\
\frac{\dd v}{\dd x} &= w \label{eqn:dv/dx},
\\
\frac{\dd s}{\dd x} &= \frac{b(x)p(x)-g(x)q(x)}{g(x)^2+b(x)^2} \label{eqn:ds/dx},
\\
\frac{\dd w}{\dd x} &= \frac{s^2}{v^3}-\frac{g(x)p(x)+b(x)q(x)}{v(g(x)^2+b(x)^2)},
\label{eqn:dw/dx}
\end{align}
\end{subequations}
with boundary condition as
\begin{align}
\theta(0)=0, \quad v(0)=1, \quad s(L)=0, \quad w(L)=0.
\label{eqn:ode_boundary}
\end{align}
Therefore, as the assessment problem in this paper, we have the TPBV problem of the nonlinear ODE \eqref{eqn:ode} with \eqref{eqn:ode_boundary}.

In this paper, we compute numerical solutions of the TPBV problem of the nonlinear ODE \eqref{eqn:ode}-\eqref{eqn:ode_boundary} with the iterative method based on \cite{Susuki1}.
As far as we have used in this and \cite{Baek,Susuki1,Susuki2,Mizuta2,Tadano}, the iterative method based on discretisation works correctly (i.e., provides physically relevant outputs of voltage profile).
Note that the authors of \cite{Susuki2} report a comparison of the computation of distribution voltage profile with the nonlinear ODE model and the standard power-flow equation.
It is shown that the distribution voltage profiles computed with the two different mathematical models are consistent for a rudimental feeder configuration.

For simplicity of the above introduction, we have assumed that except for the substation, no voltage regulator device such as load ratio control transformer and step voltage regulator is operated.
Note that it is possible to include the effect of such voltage regulation devices in the ODE model: see \ref{sec:AppC}.

\section{Asymptotic Characterisation: Theoretical}
\label{sec:perturb}

This section is devoted to theoretical studies of the TPBV problem of the nonlinear ODE \eqref{eqn:ode}.
For this, we apply the regular perturbation technique to the nonlinear ODE \eqref{eqn:ode} and derive a series of IV problems of linear ODEs.
Also, we provide the theorem of existence of solutions of the nonlinear TPBV problem, which will be the basis of our asymptotic charcterisation.
In the rest of this section, we suppose that the independent variable $x$ belongs in a closed interval $[0,L]$ where $0<L<\infty$, and that the power density functions $p(x)$ and $q(x)$ take the following forms with a small positive parameter $\varepsilon$:
\begin{equation}
p(x)=\varepsilon\tilde{p}(x),\qquad q(x)=\varepsilon\tilde{q}(x).
\label{eqn:tilde}
\end{equation}
The parameter $\varepsilon$ controls the magnitude of the impacts of loads and EVs on the distribution feeder.
Also, the position-dependent function $\tilde{p}(x)$ (or $\tilde{q}(x)$) determines the spatial shape of demand of active power (or reactive power) along the feeder and can be taken from a suitable space of functions for the current theoretical study and numerical one in Section~\ref{sec:numerics}.
For the simplicity of the theoretical analysis, it is supposed that the conductance and susceptance of the feeder are constant in $x$:
\begin{equation}
g(x)=G,\qquad b(x)=B,
\label{eqn:hoge}
\end{equation}
where $G$ and $B$ are constants and can be determined from practice.

\subsection{Application of Regular Perturbation Method}
\label{subsec:RegularPerturb}

First of all, we apply the brute-force application of the regular perturbation technique \cite{Guckenheimer} to the nonlinear ODE \eqref{eqn:ode}.
It is supposed that the solutions $\theta(x)$, $v(x)$, $s(x)$, and $w(x)$ of the TPBV problem of the nonlinear ODE \eqref{eqn:ode} with the boundary condition \eqref{eqn:ode_boundary} are expanded with polynomials of $\varepsilon$, given by
\begin{equation}
   \left.
   \begin{array}{ccl}
      \displaystyle \theta(x) &\sim& \displaystyle 0+\varepsilon \theta_1(x) +\varepsilon^2 \theta_2(x)+\varepsilon^3\theta_3(x)+\varepsilon^4\theta_4(x)+\cdots
      \\\noalign{\vskip 1mm}
      \displaystyle v(x) &\sim& \displaystyle 1+\varepsilon v_1(x)+\varepsilon^2 v_2(x)+\varepsilon^3 v_3(x)+\varepsilon^4 v_4(x)+\cdots
      \\\noalign{\vskip 1mm}
      \displaystyle s(x) &\sim& \displaystyle 0+\varepsilon s_1(x)+\varepsilon^2 s_2(x)+\varepsilon^3 s_3(x)+\varepsilon^4 s_4(x)+\cdots
      \\\noalign{\vskip 1mm}
      \displaystyle w(x) &\sim& \displaystyle 0+\varepsilon w_1(x)+\varepsilon^2 w_2(x)+\varepsilon^3 w_3(x)+\varepsilon^4 w_4(x)+\cdots
   \end{array}
   \right\},
   \label{eqn:per}
\end{equation}
where the perturbation terms $\theta_i(x)$, $v_i(x)$, $s_i(x)$, and $w_i(x)$ ($i=1,2,\ldots$) are functions defined on the closed interval $[0,L]$ with the following conditions at $x=0$ and $x=L$:
\begin{equation}
   \left.
   \begin{array}{ccccccl}
      \displaystyle \theta_1(0) &=& \theta_2(0) &=& \cdots &=& 0
      \\\noalign{\vskip 1mm}
      \displaystyle v_1(0) &=& v_2(0) &=& \cdots &=& 0
      \\\noalign{\vskip 1mm}
      \displaystyle s_1(L) &=& s_2(L) &=& \cdots &=& 0
      \\\noalign{\vskip 1mm}
      \displaystyle w_1(L) &=& w_2(L) &=& \cdots &=& 0
   \end{array}
   \right\}.
   \label{eqn:ode0}
\end{equation}
Note that the degree of regularity of solutions of the TPBV problem with respect to $\varepsilon$ is hard to know from the generality of the choice of $\tilde{p}(x)$ and $\tilde{q}(x)$.
The above expansion is a formalism, and its justification remains to be solved.

Thus, the associated linear ODEs for each order of the perturbation terms are derived.
By substituting \eqref{eqn:per} into \eqref{eqn:ode} with \eqref{eqn:tilde} and \eqref{eqn:hoge}, and picking up the coefficients of the small parameter $\varepsilon$, the following linear ODEs for the 1st-order perturbation terms $\theta_1 (x)$, $v_1 (x)$, $s_1 (x)$, and $w_1 (x)$ are derived:
\vspace{0mm}
\begin{equation}
\left.
\begin{array}{ccl}
\displaystyle\frac{\dd \theta_1}{\dd x}&=&\displaystyle-s_1
\\\noalign{\vskip 2mm}
\displaystyle\frac{\dd v_1}{\dd x}&=&\displaystyle w_1
\\\noalign{\vskip 2mm}
\displaystyle\frac{\dd s_1}{\dd x}&=&\displaystyle\frac{B\tilde{p}(x)-G\tilde{q}(x)}{G^2+B^2}
\\\noalign{\vskip 2mm}
\displaystyle\frac{\dd w_1}{\dd x}&=&\displaystyle-\frac{G\tilde{p}(x)+B\tilde{q}(x)}{G^2+B^2}
\end{array}
\right\}.
\label{eqn:per1}
\end{equation}
That is, we have the IV problem of the linear ODE \eqref{eqn:per1} with the initial values \eqref{eqn:ode0}.
This problem is self-consistent in the sense that except for the unknown functions $\theta_1 (x)$, $v_1 (x)$, $s_1 (x)$, and $w_1 (x)$, all the parameters are given.
In the similar manner as above, the linear ODEs for the 2nd-, 3rd-, and 4th-order perturbation terms are derived as follows:
\begin{equation}
   \left.
   \begin{array}{ccl}
      \displaystyle \frac{\dd \theta_2}{\dd x}&=& \displaystyle 2v_1(x)s_1(x)
      \\\noalign{\vskip 2mm}
      \displaystyle \frac{\dd v_2}{\dd x}&=&w_2
      \\\noalign{\vskip 2mm}
      \displaystyle \frac{\dd s_2}{\dd x}&=&0
      \\\noalign{\vskip 2mm}
      \displaystyle \frac{\dd w_2}{\dd x}&=&
      s_1(x)^2+\displaystyle\frac{G\tilde{p}(x)+B\tilde{q}(x)}{G^2+B^2}v_1(x)
   \end{array}
   \right\},
   \label{eqn:per2}
\end{equation}
\begin{equation}
   \left.
   \begin{array}{ccl}
      \displaystyle \frac{\dd \theta_3}{\dd x} &=&
      4v_1(x)^2s_1(x)+\{2v_2(x)+v_1(x)^2\}s_1(x)
      \\\noalign{\vskip 2mm}
      \displaystyle \frac{\dd v_3}{\dd x} &=& w_3
      \\\noalign{\vskip 2mm}
      \displaystyle \frac{\dd s_3}{\dd x} &=& 0
      \\\noalign{\vskip 2mm}
      \displaystyle \frac{\dd w_3}{\dd x} &=&
      -3s_1(x)^2v_1(x)
      + \displaystyle\frac{G\tilde{p}(x)+B\tilde{q}(x)}{G^2+B^2}{
      v_2(x)}
   \end{array}
   \right\},
   \label{eqn:per3}
\end{equation}
and
\begin{equation}
   \left.
   \begin{array}{ccl}
      \displaystyle \frac{\dd\theta_4}{\dd x} &=&
      -2v_1(x) 
      \left[4v_1(x)^2s_1(x)+\{2v_2(x)+v_1(x)^2\}s_1(x)\right] \\
      & &
      -\{2v_2(x)+v_1(x)^2\} 
      \cdot 2v_1(x)s_1(x) \\
      & &
      -2\{v_3(x)+v_1(x) v_2(x)\} 
      \cdot (-s_1(x))
      \\ 
      \displaystyle \frac{\dd v_4}{\dd x} &=& w_4
      \\\noalign{\vskip 2mm}
      \displaystyle \frac{\dd s_4}{\dd x} &=& 0
      \\\noalign{\vskip 2mm}
      \displaystyle \frac{\dd w_4}{\dd x} &=&
      -3s_1(x)^2v_2(x)
      +\displaystyle\frac{G\tilde{p}(x)+B\tilde{q}(x)}{G^2+B^2}{v_3(x)}
   \end{array}
   \right\}.
   \label{eqn:per4}
\end{equation}
Namely, we have the IV problems of the linear ODEs \eqref{eqn:per2}, \eqref{eqn:per3}, and \eqref{eqn:per4} with the initial values \eqref{eqn:ode0}.
The $i$-th order IV problem becomes self-consistent if all the problems with order lower than the $i$-th have unique solutions with appropriate regularity.
This implies that it is possible to determine the perturbation terms in a recursive manner from the 1st order problem.
Also, since all the ODEs are linear, it is possible to derive analytical forms of the solutions that are direct to assessing the impacts of $\tilde{p}(x)$ and $\tilde{q}(x)$ on the 
variables such as the voltage amplitude $v(x)$, which is one of the benefits of our ODE approach from the technological viewpoint.

\subsection{Existence Results of Solutions}
\label{subsec:theory}

Here we collect a series of theoretical results on existence and regularity of solutions for the nonlinear TPBV problem and on existence, uniqueness, and regularity of solutions for the derived linear IV problems.
The notation $C^r[0,L]$ used below represents the space of $r$-times differentiable functions defined on the finite, closed-interval $[0,L]$.
The case $r=0$ implies the space of continuous functions on $[0,L]$

Before applying the perturbation technique, it is the first study to prove that the original nonlinear problem, namely, the nonlinear TPBV problem, has a solution under presence of the perturbation terms.
\begin{theorem}
\label{thm:TPBV}
Consider the TPBV problem of the nonlinear ODE \eqref{eqn:ode} with the boundary condition \eqref{eqn:ode_boundary}, and assume $\tilde{p}(x),\tilde{q}(x)\in C^0[0,L]$.
Then, there exists a constant $\varepsilon>0$ such that the problem has a solution with $C^1$ regularity.
\end{theorem}
\begin{proof}
See \ref{sec:AppA}.
\end{proof}
\begin{remark}
The $C^1$ regularity of solutions validates the application of finite-difference scheme \cite{Keller} for locating them numerically.
We use the so-called central finite-difference scheme \cite{Keller} in this paper.
\end{remark}

Next, we consider the derived  IV problems of the linear ODEs \eqref{eqn:per1} to \eqref{eqn:per4}.
The existence and uniqueness of solutions of the problems are a simple outcome of applying the standard theory of ODEs \cite{Hale}.
For this, let us denote the $n$-th order problems described by \eqref{eqn:per1} to \eqref{eqn:per4} as follows: for $\bm{u}_{n}=(\theta_n,v_n,s_n,w_n)^{\top}$ $(n=1,\ldots,4)$,
\begin{equation}
  \frac{\dd\bm{u}_{n}}{\dd x} = \bm{F}_{n}(x,\bm{u}_{1},\bm{u}_{2},\bm{u}_{3},\bm{u}_{4}),
  \label{eqn:nthode}
\end{equation}
with the boundary conditions \eqref{eqn:ode0}.
The $\bm{F}_n$ describe the right-hand sides of \eqref{eqn:per1} to \eqref{eqn:per4}.
The following theorem with $C^1$ regularity now holds.
\begin{theorem}
\label{thm:LINEAR}
Consider the IV problem of the $n$-th linear ODE for $n=1,\ldots,4$, and assume $\tilde{p}(x),\tilde{q}(x)\in C^0[0,L]$.
Then, the problem has a unique solution $\bm{u}_n(x)\in C^1([0,L]^4)$.
\end{theorem}
\begin{proof}
First, consider the 1st order problem.
Since $\tilde{p}(x),\tilde{q}(x)\in C^0[0,L]$ is assumed, by direct integration of the right-hand sides of \eqref{eqn:per1} for $\dd s_1/\dd x$ and $\dd w_1/\dd x$, the solutions $s_1(x)$ and $w_1(x)$ are unique and in $C^1[0,L]$.
Then, from the ODEs for $\theta_1$ and $v_1$ in \eqref{eqn:per1}, the solutions $\theta_1(x)$ and $v_1(x)$ are in $C^1[0,L]$.
In the same manner, the higher order problems are also solvable recursively and have unique solutions in $C^1[0,L]$.
\end{proof}

In addition to this, regarding the voltage amplitude $v_n$ (and voltage gradient $w_n$), it is possible to state a stronger result because the following equations of the derivatives $w_n$ and $w_n$ are derived for \emph{all} $n=3,4,\ldots$:
\begin{equation}
\frac{\dd v_n}{\dd x} = w_n, \qquad
\frac{\dd w_n}{\dd x} = -3s_1(x)^2v_{n-2}(x)+\frac{G\tilde{p}(x)+B\tilde{q}(x)}{G^2+B^2}v_{n-1}(x).
\label{eqn:v34}
\end{equation}
Thus, we have the following theorem that is fundamental for the asymptotic charcterisation of distribution voltage profile:
\begin{theorem}
\label{thm:v-w}
Consider the IV problem of the $n$-th linear ODE \eqref{eqn:v34} with $v_n(0)=0$ and $w_n(L)=0$ for $n=3,4,\ldots$, and assume $\tilde{p}(x), \tilde{q}(x)\in C^0[0,L]$.
Then, the problem has a unique solution $(v_n(x),w_n(x))^\top\in C^1([0,L]^2)$.
\end{theorem}
\begin{proof}
The proof is almost the same as Theorem~\ref{thm:LINEAR}.
Since $p(x),q(x)\in C^0[0,L]$ and $s_1(x),v_1(x),v_2(x)\in C^1[0,L]$ hold, by integration of the right-hand sides of \eqref{eqn:v34}, we immediately see $w_n(x)\in C^1[0,L]$ and also $v_n(x)\in C^1[0,L]$ at least.
\end{proof}
\begin{remark}
For the rest of the independent variables, namely $\theta_n(x)$ and $s_n(x)$, we do not have the result parallel to Theorem~\ref{thm:v-w} unfortunately because no equation parallel to \eqref{eqn:v34} is derived for $\theta_n(x)$ and $s_n(x)$.
\end{remark}

%
%

\subsection{Remarks}

\subsubsection{Physical Meaning of the Regular Perturbation}

In Section~\ref{subsec:RegularPerturb}, we introduced the small positive parameter $\varepsilon$ and associated asymptotic expansion.
Here, let us introduce a physical interpretation of $\varepsilon$ in order to make it clear to show its utility.
The parameter controls the magnitude of the impacts of loads and EVs on the distribution feeder.
Technically and interestingly, it can be interpreted as the ratio of demand's utilization (in kW) with respect to the capacity (also in kW) of a distribution transformer.
This ratio is related to the congestion management of distribution grids \cite{Arias:2019} in which overloading distribution transformers is taken into consideration as in \cite{NordicOutlook}.
Now, we decompose $\varepsilon$ into the two parameters: the contribution by EV charging, denoted by $\varepsilon_{\rm ev}$, and the ontribution by the other loads and DERs by $\varepsilon_{\rm load}$:
\begin{equation}
  \varepsilon=\varepsilon_{\rm ev}+\varepsilon_{\rm load},
  \label{eqn:ep_EV}
\end{equation}
where we will see that the two parameters are bounded above by 1.
Accordingly, the power density functions $p(x)$ and $q(x)$ are rewritten as follows:
\begin{equation}
  p(x)=\varepsilon_{\rm ev}\tilde{p}(x)+\varepsilon_{\rm load}\tilde{p}(x),\qquad q(x)=\varepsilon_{\rm ev}\tilde{q}(x)+\varepsilon_{\rm load}\tilde{p}(x).
\end{equation}
Then, from \eqref{eqn:per}, the voltage amplitude $v(x)$ is expanded in terms of $\varepsilon_{\rm ev}$ and $\varepsilon_{\rm load}$ as
\begin{align}
  v(x) &= 1+\varepsilon_{\rm load} v_1(x)+{\varepsilon^2_{\rm load}}v_2(x)+{\varepsilon^3_{\rm load}}v_3(x)+\cdots
  \nonumber\\
   & \quad+\varepsilon_{\rm ev} v_1(x)+{\varepsilon^2_{\rm ev}}v_2(x)+{\varepsilon^3_{\rm ev}} v_3(x)+\cdots
   \nonumber\\
   & \quad +2\varepsilon_{\rm ev}\varepsilon_{\rm load}v_2(x)+3\varepsilon_{\rm ev}\varepsilon_{\rm load}(\varepsilon_{\rm ev}+\varepsilon_{\rm load})v_3(x)+\cdots.
  \label{eqn:v_ep}
\end{align}
To explain each of the lines on the right-hand side, let us consider a loading condition of the distribution feeder where the loads except for EVs are originally connected and known a prior.
If no EV is connected, that is $\varepsilon_{\rm ev}=0$, then the voltage amplitude is completely evaluated with the first line.
Thus, the impact of EV charging \emph{conditioned} by the loads is quantified with the second and third lines, defined as
\begin{align}
  \Delta v_{\rm ev\,|\,load}(x) &:=
  \varepsilon_{\rm ev} v_1(x)+{\varepsilon^2_{\rm ev}}v_2(x)+{\varepsilon^3_{\rm ev}}v_3(x)
  +\cdots \nonumber
  \\
  &\quad +2\varepsilon_{\rm ev}\varepsilon_{\rm load}v_2(x)
  +3\varepsilon_{\rm ev}\varepsilon_{\rm load}(\varepsilon_{\rm ev}+\varepsilon_{\rm load})v_3(x)
  +\cdots.
  \label{eqn:v_ev}
\end{align}
Here, it should be noted that the perturbation terms $v_n(x)$ are determined solely by the \emph{scaled} power density functions $\tilde{p}(x)$ and $\tilde{q}(x)$ not the ratio $\varepsilon_{\rm ev}$, and that $\tilde{p}(x)$ and $\tilde{q}(x)$ represent the spatial shape of demand of power that is determined mainly by locations of loading centers and charging points.
This implies that \eqref{eqn:v_ev} provides a simple method for quantifying how the EV penetration affects the voltage amplitude, which will be demonstrated in Figure~\ref{fig:v_EV} and Table~\ref{tab:v_ev}.
It becomes realized for the first time using the ODE approach and thus shows its technological benefit in comparison with the conventional method based on power-flow equations, which requires inevitable iterative computations of nonlinear programming.

\begin{figure}[!t]
  \centering
  \includegraphics[width=0.4\hsize]{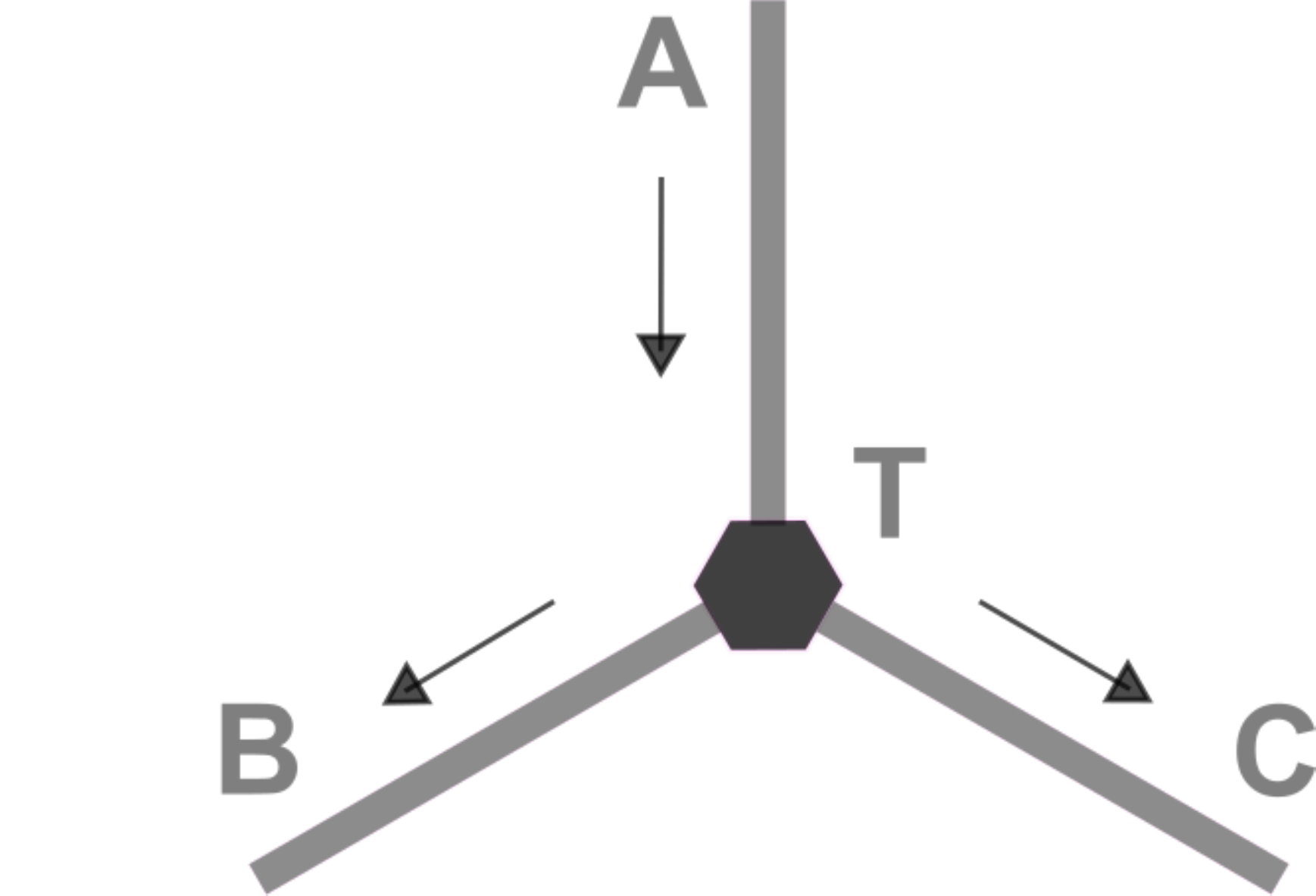}
  \caption{Basic configuration of three feeders with one bifurcation point.
  Three branch lines  A, B, and C are connected to the bifurcation point T. The arrows respect the reference directions of current flows.}
  \label{fig:branch}
\end{figure}

\subsubsection{Representation of Feeder's Bifurcation}

In assessment of realistic configurations, it is inevitable to consider the bifurcation of distribution feeder in our framework of asymptotic expansion.
To take it into account, we consider the three feeders with one bifurcation point shown in Figure~\ref{fig:branch}.
For this case, in order to formulate the nonlinear TPBV problem, we introduce as in \cite{Susuki1,Susuki2} an additional boundary condition at the bifurcation point.
It is shown in \cite{Susuki1,Susuki2} that at this point, the voltage phase $\theta$ and voltage amplitude $v$ are continuous along any pair of the three feeders, labeled as A, B, and C, while the auxiliary variable $s$ and voltage gradient $w$ are not continuous.
These are described as follows:
\begin{equation}
   \left.
   \begin{array}{l}
      \theta_{\rm A} = \theta_{\rm B} = \theta_{\rm C}\\
      v_{\rm A} = v_{\rm B} = v_{\rm C}\\
      s_{\rm A} = s_{\rm B} + s_{\rm C}\\
      w_{\rm A} = w_{\rm B} + w_{\rm C}
   \end{array}
   \right\},
\label{eqn:boundary}
\end{equation}
where $\theta\sub{A}$, $v\sub{B}$, and so on represent the values of the dependent variables taken as limits to the bifurcation point along feeders A, B, and C, respectively\footnote{
For example, if feeder A corresponds to that in Figure~\ref{fig:simplefeeder} and is connected to be the bifurcation point at $x=L$, then $\theta_{\rm A}$ is defined as $\displaystyle\lim_{x\to L} \theta(x)$.}.
It should be noted that the boundary conditions for $s$ and $w$ are dependent on the choice of reference directions of current flows: \eqref{eqn:boundary} holds for the directions in Figure~\ref{fig:branch}.

Now, we show a set of boundary conditions associated with \eqref{eqn:boundary} in the asymptotic expansion.
For this, using the three different independent variables $x_{j}\in\mathbb{R}$ (${j}\in\{\rm A, B, C\}$) for the three feeders, we suppose the following expansions of the functions $\theta_{j}(x_{j})$, $v_{j}(x_{j})$, $s_{j}(x_{j})$, and $w_{j}(x_{j})$ around a neighborhood of the bifurcation point (represented in the $x_j$-coordinates as $x_{j}=0$) in terms of a common small parameter $\varepsilon$:
\begin{equation}
\left.
\begin{aligned}
\theta_{j}(x_{j})\sim \sum^\infty_{i=1}\varepsilon^i \theta_{j,i}(x_{j}), \quad&
v_{j}(x_{j})\sim 1+\sum^\infty_{i=1}\varepsilon^i v_{j,i}(x_{j}) \quad\\
s_{j}(x_{j})\sim \sum^\infty_{i=1}\varepsilon^i s_{j,i}(x_{j}), \quad &
w_{j}(x_{j})\sim \sum^\infty_{i=1}\varepsilon^i w_{j,i}(x_{j}) \quad
\end{aligned}
\right\}.
\end{equation}
Thereby, the condition \eqref{eqn:boundary} at the bifurcation point is re-written in the asymptotic framework as follows:
\begin{equation}
   \left.
   \begin{array}{ccccl}
      \displaystyle \theta_{\rm A1} = \theta_{\rm B1} = \theta_{\rm C1}, && \theta_{\rm A2} = \theta_{\rm B2} = \theta_{\rm C2}, &\quad& \cdots \qquad
      \vspace{1mm}
      \\\displaystyle v_{\rm A1} = v_{\rm B1} = v_{\rm C1}, &\quad& v_{\rm A2} = v_{\rm B2} = v_{\rm C2}, &\quad& \cdots \qquad
      \vspace{1mm}
      \\\displaystyle s_{\rm A1} = s_{\rm B1} + s_{\rm C1}, &\quad& s_{\rm A2} = s_{\rm B2} + s_{\rm C2}, &\quad& \cdots \qquad
      \vspace{1mm}
      \\\displaystyle w_{\rm A1} = w_{\rm B1} + w_{\rm C1}, &\quad& w_{\rm A2} = w_{\rm B2} + w_{\rm C2}, &\quad& \cdots \qquad
   \end{array}
   \right\}
   \label{eqn:perbound}
\end{equation}
where $\theta_{{j},i}$ is defined as $\displaystyle \lim_{x_{j}\to 0}\theta_{{j},i}(x)$ and so on.
The derivation of \eqref{eqn:perbound} is presented in \ref{sec:AppB}.
This shows that it is possible to perform a low-order approximation of the distribution voltage profile in a self-consistent manner over the bifurcation point.
This will be used in Section~\ref{sec:numerics}\ref{subsec:theory} for numerical simulations.

Furthermore, it is practically inevitable to consider the case where a voltage regulation device (Step Voltage Regulator \cite{Kersting})
is installed in the feeder.
In this case, an additional boundary condition can be formulated in the similar manner as above and \cite{Susuki2}.
This is summarized in \ref{sec:AppC} for wider utility of the asymptotic assessment.

\section{Numerical Demonstration}
\label{sec:numerics}

This section is devoted to numerical demonstration of the asymptotic charcterisation of distribution voltage profile 
in Section~\ref{sec:perturb}.
The demonstration is done with numerical solutions of the nonlinear TPBV problem for two distribution models.
The main idea for the demonstration is to compare the asymptotic expansions \eqref{eqn:per} up to the 1st, 2nd, 3rd, and 4th perturbation terms with direct numerical solutions for the simple feeder model (see Figure~\ref{fig:simple}) and the practical model (see Figure~ \ref{fig:kanden}).
The detailed setting of the numerical demonstrations in this section is presented in \ref{sec:AppD}.

Here, in practical situations, the power demand and generation along a feeder happen in a discrete manner.
To precisely state this, we suppose that $N$ number of load and stations are located at $x=\xi_i\in(0,L)$ ($i=1,\ldots,N$) satisfying $\xi_{i+1}<\xi_i$.
Then, by denoting as $P_i$ the active power consumed at $x=\xi_i$, the power density functions $p(x)$ is given as follows:
\begin{equation}
   p(x)
   = \sum^{N}_{i=1} P_i \delta(x-\xi_i),
\label{eqn:pq}
\end{equation}
where $\delta(x-\xi_i)$ is the Dirac's delta-function supported at $x=\xi_i$.
The formulation of $p(x)$ is excluded in the theoretical development of Section~\ref{sec:perturb}\ref{subsec:theory} because of the $C^0$ assumption and, furthermore, makes it difficult to numerically approximate solutions of the nonlinear ODE.
To avoid these, as in \cite{Susuki2,Mizuta2}, we use the following coarse-graining of $p(x)$ with the Gaussian function:
\begin{equation}
   p(x)\sim\sum_{i=1}^N \frac{P_i}{\sqrt{2\pi\sigma^2}}\mathrm{exp}\biggl(-\frac{(x-\xi_i)^2}{2\sigma^2}\biggr),
   \label{eqn:Gauss}
\end{equation}
where we regard $x,\xi_i$ as scalars, and $\sigma^2$ is the variance.
The coarse-grained $p(x)$ is clearly $C^1$, can be treated as in Section~\ref{sec:perturb}\ref{subsec:theory}, and hence lead to the existence of solutions for the nonlinear TPBV problem.
The parameter $\sigma(>0)$ is fixed at a constant distance, which is sufficiently smaller than the interval between loads or EV charging stations.

\subsection{Simple Feeder}
\begin{figure}[t]
   \centering
   \includegraphics[width=0.7\hsize]{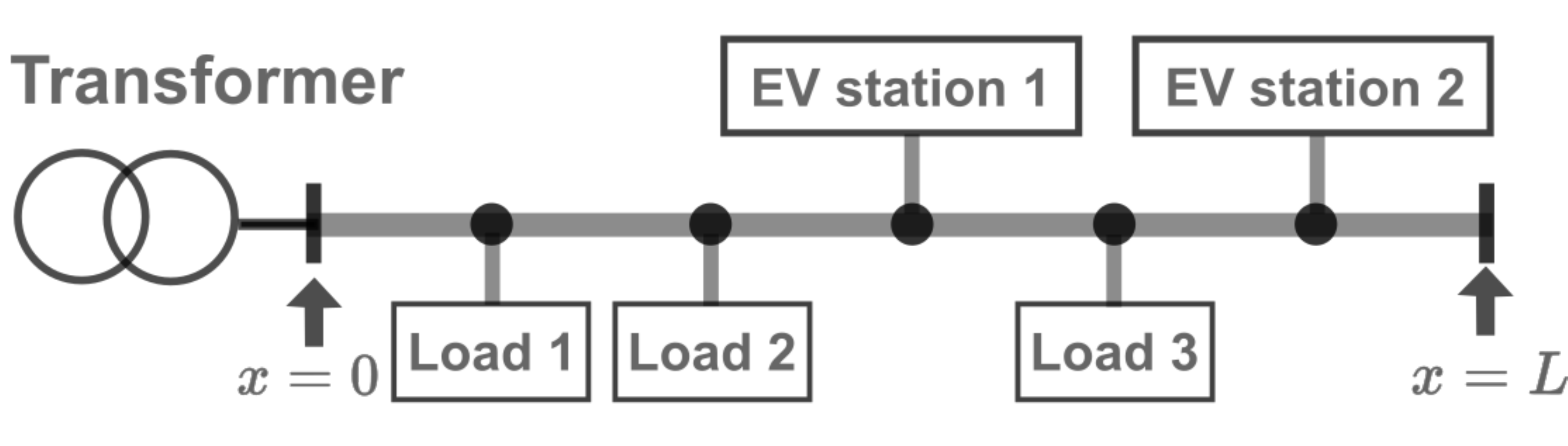}
   \caption{Simple distribution feeder model.
   The feeder has the 3 loads and 2 EV stations located at a common interval (0.5\,km).}
   \label{fig:simple}
\end{figure}

First, we evaluate the asymptotic expansion
\eqref{eqn:per} for the simple feeder model in Figure~\ref{fig:simple}.
Specifically, the three functions $\theta(x)$, $v(x)$, and $w(x)$ are addressed; $s(x)$ is not addressed because it is equivalent to the 1st-order $\varepsilon s_1(x)$, and the higher-order terms $s_i(x)$ are identically zero.
The feeder model in Figure~\ref{fig:simple} possesses the 3 loads and 2 EV charging stations located at a common interval (0.5\,km).
We assume that the rated capacity of the transformer is set at 12\,MVA, and no loads exists at the end point of the feeder for simplicity of the analysis.
It is also assumed that all the loads are connected via inverters like EVs and thus operated under unity power-factor mode.
This implies that the reactive power $q(x)$ is here negligible, namely $q(x)=0$ for all $x\in[0{\rm km},5{\rm km}]$.
All the loads connected to the feeder are constant in time.
The power consumption for each of EV station 1 and 2 is 0.200\,pu,
the power consumption for each of Load 1, 2, and 3 is 0.133\,pu, and therefore the total amount of the power consumption is set as 80\,\% of the transformer's capacity.
The coarse-grained $p(x)$ for simulations of the nonlinear and linear ODEs is shown in the top of Figure~\ref{fig:ode_sim}.
\begin{figure}[h]
   \centering
      \includegraphics[width=0.6\textwidth]{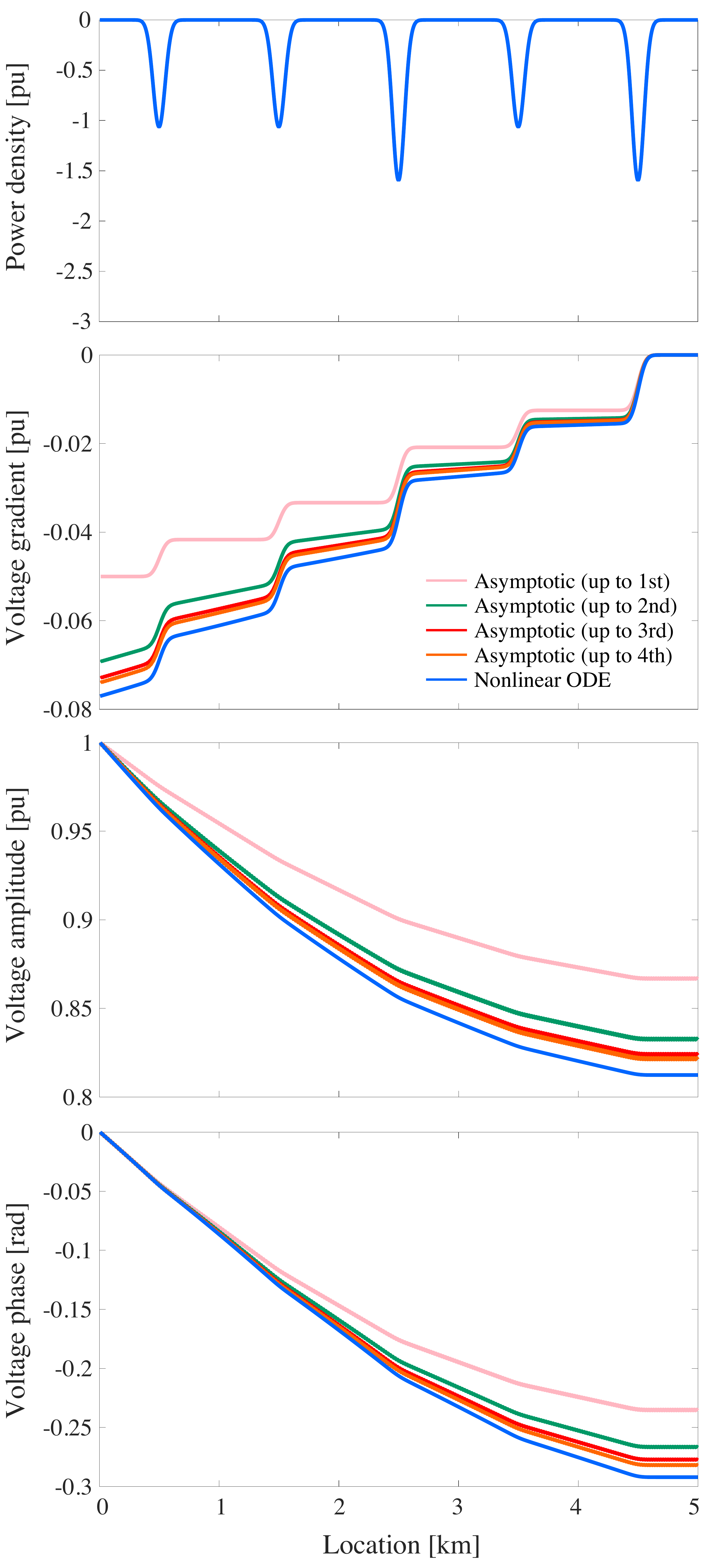}
      \caption{%
      Validation of proposed asymptotic expansions
      \eqref{eqn:per} for simple distribution feeder model in~Figure \ref{fig:simple}.
      The power density function $p(x)$ and associated numerical solutions of the nonlinear ODE model \eqref{eqn:ode} are also shown.
      }
      \label{fig:ode_sim}
\end{figure}
\begin{table}[t]
   \centering
   \caption{Validation of asymptotic expansion \eqref{eqn:per} by differences from nonlinear ODE for simple distribution feeder model in Figure~\ref{fig:simple}}
   \label{tab:sim}
   \vskip 2mm
   \begin{tabular}{lccc}\hline
        & $\Delta w(0\,{\rm km})$
        & $\Delta v(L=5\,{\rm km})$ & $\Delta \theta(L=5\,{\rm km})$\\\hline%
      asymptotic expansion up to 1st-order
      & 0.02700 & 0.0542 & 0.0564 \\\hline
      asymptotic expansion up to 2nd-order
      & 0.00951 & 0.0191 & 0.0250 \\\hline
      asymptotic expansion up to 3rd-order
      & 0.00580 & 0.0102 & 0.0145 \\\hline
      asymptotic expansion up to 4th-order
      & 0.00478 & 0.0078 & 0.0099 \\\hline
   \end{tabular}
\end{table}

Figure~\ref{fig:ode_sim} and Table~\ref{tab:sim} show the proposed asymptotic expansions and direct numerical solutions of the nonlinear ODE with the power density function in the top of Figure~\ref{fig:ode_sim}.
The choice of the value of $\varepsilon$ is an issue for numerical simulations of the asymptotic expansion because it is not guided by perturbation theory.
Here, for a fixed $p(x)$ we set $\varepsilon$ at multiple small values for simulations\footnote{This implies in \eqref{eqn:Gauss} with $P_{i}=\varepsilon\tilde{P}_{i}$ that we tune both values of $\varepsilon$ and $\tilde{P}_{i}$ while keeping $P_i$.}, and the following description is consistent for $\varepsilon=1\times10^{-1},1\times 10^{-2},\ldots,1\times10^{-8}$.
The second row of Figure~\ref{fig:ode_sim} shows the computational results on voltage gradient $w(x)$:
the proposed asymptotic expansions
up to 1st- to 4th-order terms by \emph{pink, green, red, and orange solid} lines, and the direct numerical solution of the nonlinear ODE by \emph{blue dashed} line.
The difference between each of the asymptotic expansions and the nonlinear ODE increases from the end to the start of the feeder due to the effect of power consumption by loads.
Similarly, the associated voltage amplitude $v(x)$ and voltage phase $\theta(x)$ are shown in the third row and bottom of Figure~\ref{fig:ode_sim}.
For both, the difference between each of the asymptotic expansions and the nonlinear ODE increases from the start to the end of the feeder.
It is clearly shown in Figure~\ref{fig:ode_sim} that the asymptotic expansions for $w(x)$, $v(x)$, and $\theta(x)$ approach to the direct numerical solution of the nonlinear ODE as the order increases.
This is confirmed from the quantification of  differences between the asymptotic expansion and the nonlinear ODE for the computed values, denoted as $\Delta w(0{\rm km})$, $\Delta v(L=5{\rm km})$, and $\Delta \theta(L=5{\rm km})$, in Table~\ref{tab:sim}.

\begin{figure}
  \centering
  \includegraphics[width=0.65\hsize]{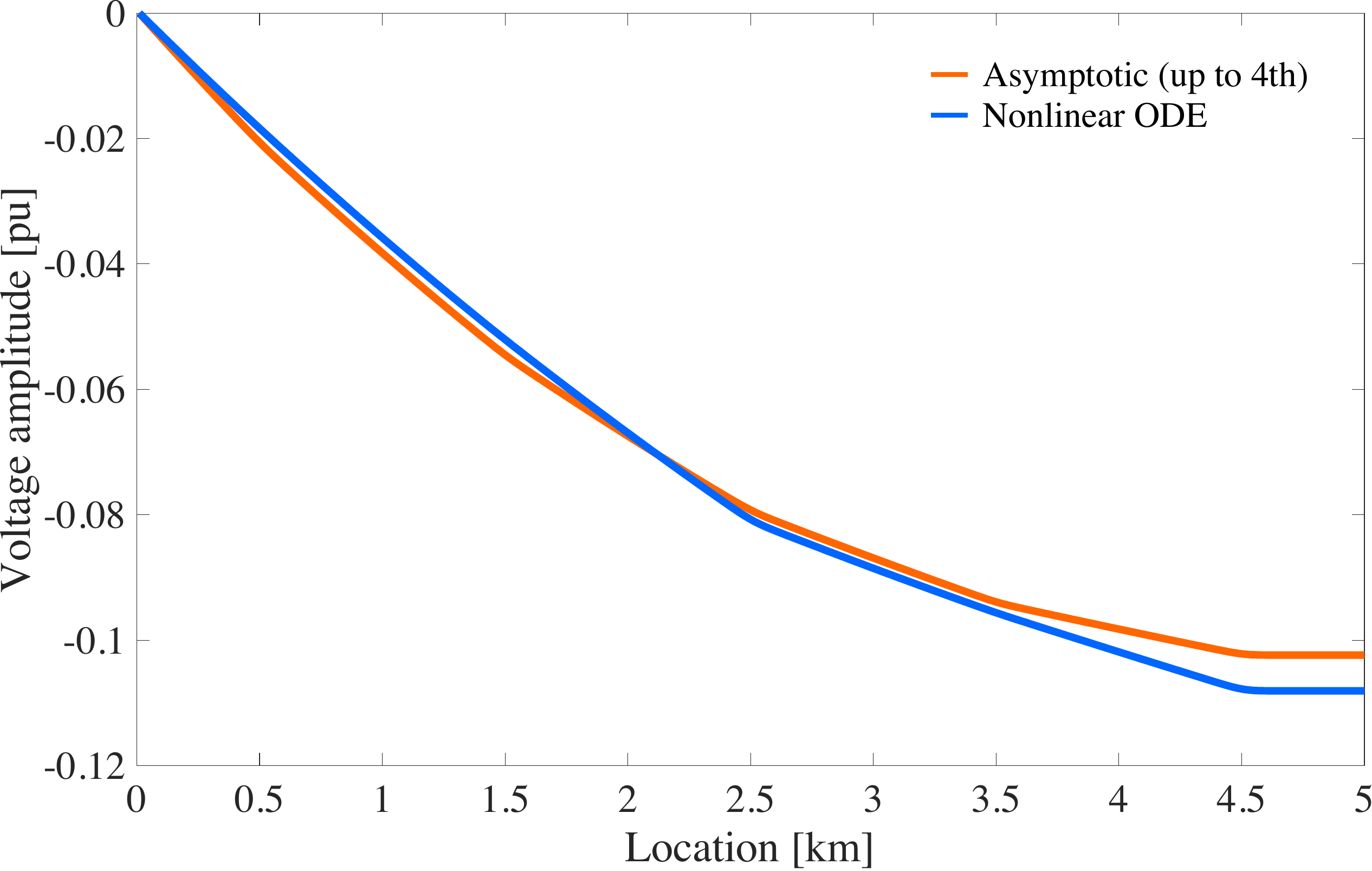}
  \caption{Impact assessment of the EV charging using \eqref{eqn:per}. The \emph{orange} line shows the assessment result of \eqref{eqn:per} under the asymptotic expansion up to 4th-order term and $\varepsilon_{\rm ev}=0.4\varepsilon$. The \emph{blue} line shows the associated numerical result of the nonlinear ODE.}
  \label{fig:v_EV}
\end{figure}
\begin{table}[t]
   \centering
   \caption{Parameter dependence of accuracy of the impact assessment using \eqref{eqn:per}. The error $\Delta v(L=5\,{\rm km})$ of results between \eqref{eqn:per} and the nonlinear ODE is computed at $x=L=5\,{\rm km}$. The asymptotic expansion up to 4th is used here.
   }
   \label{tab:v_ev}
   \vskip 2mm

   \begin{tabular}{cccc}\hline
      $\varepsilon_{\rm ev}/\varepsilon$ &
      Error $\Delta v(L=5{\rm km})$ \\\hline
      $30\%$
      & 0.004285  \\\hline
      $40\%$
      & 0.005704 \\\hline
      $50\%$
      & 0.01795  \\\hline
      $60\%$
      & 0.03283  \\\hline
   \end{tabular}
 \end{table}
Here, the utility of the impact assessment \eqref{eqn:per} is demonstrated.
Figure~\ref{fig:v_EV} and Table~\ref{tab:v_ev} show numerical results on the impact assessment for the single feeder model.
Specifically, in Figure~\ref{fig:v_EV} we consider the original (a prior) case where the loads (Load 1, 2, and 3) are originally connected; then we introduce the EV charging (EV stations 1 and 2). For this, by the \emph{orange } line we plot \eqref{eqn:per} under the asymptotic expansion (up to 4th-order term) and $\varepsilon_{\rm ev}=0.4\varepsilon$ (hence $\varepsilon_{\rm load}=0.6\varepsilon$), and by the \emph{blue} line we also plot the numerical simulation of the nonlinear ODE.
It is clearly shown that these numerical results are similar. Table~\ref{tab:v_ev} shows the $\varepsilon_{\rm ev}$-dependence of accuracy of the impact assessment using \eqref{eqn:per}.
The errors of results between \eqref{eqn:per} and the nonlinear ODE are small for different choices of $\varepsilon_{\rm ev}$.
The accuracy slightly decreases as $\varepsilon_{\rm ev}$ increases.
The EV station 2 is placed close to the end of feeder and thus dominantly affects the voltage profile (if it extracts current flow from the bank through the feeder). In this, the increase of $\varepsilon_{\rm ev}$ implies that the accuracy of asymptotic expansion up to 4th tends to deteriorate. This can be improved by adding higher-order terms to the assessment.
Thus, the effectiveness of the asymptotic expansion for the impact assessment is confirmed.

Consequently, the proposed asymptotic expansion is capable of evaluating the distribution voltage profile for the simple distribution feeder model in Figure~\ref{fig:simple}.

\subsection{Practical Feeder}
\begin{figure}[t]
   \centering
   \includegraphics[width=0.7\hsize]{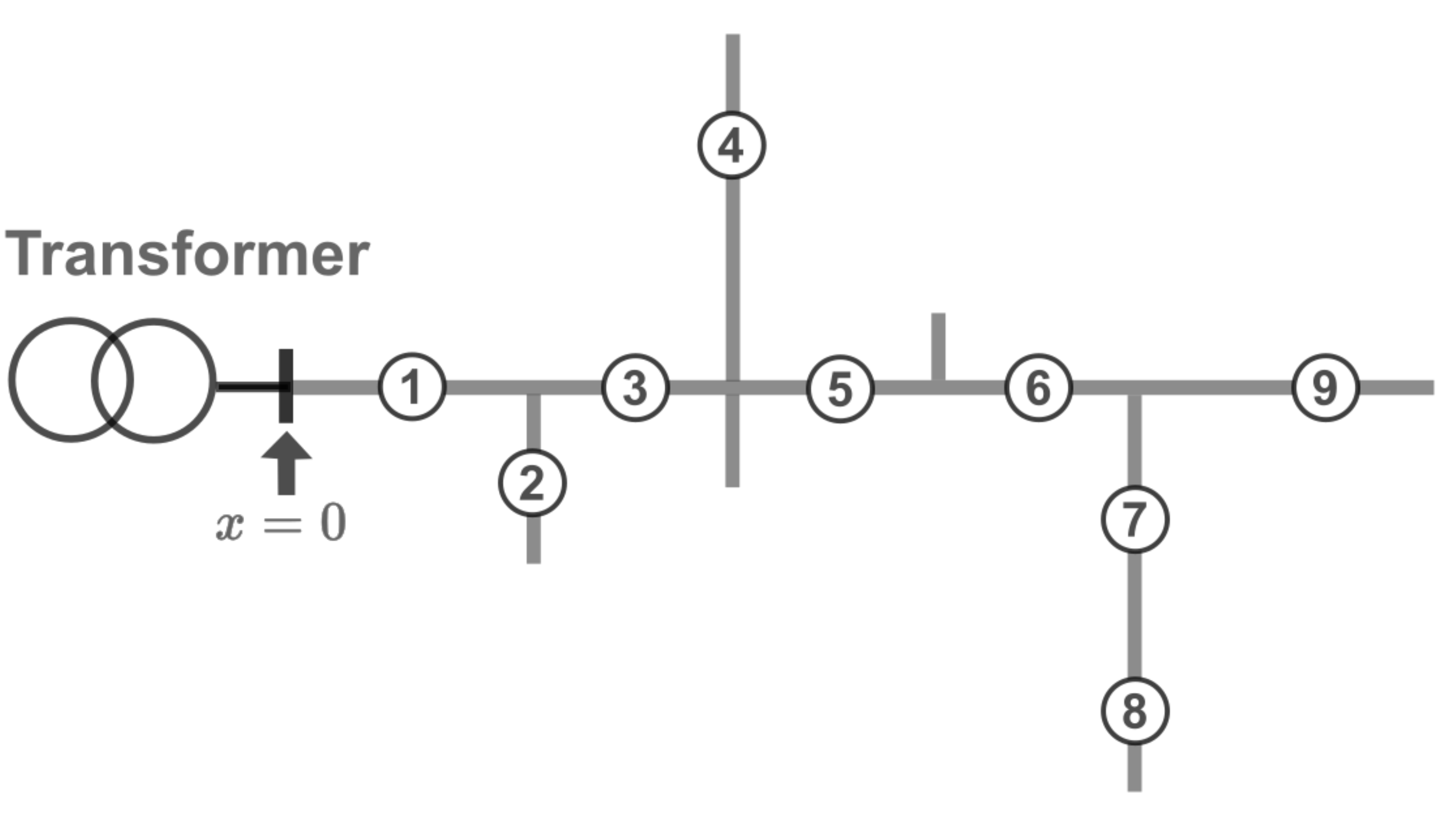}
   \caption{
   Model of multiple feeders
   based on a practical distribution grid in residential area in Japan.
   The 9 EV charging stations virtually installed and denoted by
   \emph{circled numbers}.
   }
   \label{fig:kanden}
\end{figure}

Second, we evaluate the the asymptotic expansions \eqref{eqn:per} for the practical configuration with multiple feeders and bifurcations shown in Figure~\ref{fig:kanden}.
As in the previous sub-section, for comparison of phase $\theta$, voltage amplitude $v$, and voltage gradient $w$, we consider the the proposed asymptotic expansions \eqref{eqn:per} including the perturbation terms up to 1st to 4th order and the numerical solution of the nonlinear ODE \eqref{eqn:ode}.
The model in Figure~\ref{fig:kanden} is based on a practical distribution feeder of residential area in western Japan and provided by an utility company.
A similar distribution model is used in \cite{Mizuta2} and thus is summarized in \ref{sec:AppD}.


The corresponding power density function is shown in the top of Figure~\ref{fig:sim_kanden}.
The function is constructed in the same way as in the single-feeder model.
In the figure, the \emph{positiveness} implies the \emph{discharging} operation by in-vehicle batteries, and the \emph{negativeness} does their \emph{charging} operation or the power consumption by loads.
The \emph{blue} part on the feeders represents the locations connected to the residential loads through the pole transformers.
In addition to the loads, we assume that each station has 9 EVs for simultaneous charging, where each EV has the rated charging power of 4\,kVA based on \cite{Clement1}.

Figure~\ref{fig:sim_kanden} shows the difference the asymptotic expansion \eqref{eqn:per} up to 4th and the numerical solution of the nonlinear ODE \eqref{eqn:ode} incorporated with the power density function based on the top of Figure~\ref{fig:sim_kanden}.
The simulations for the multiple feeders with bifurcations were performed with the boundary condition \eqref{eqn:boundary}.
The voltage amplitude $v(x)$ and voltage phase $\theta(x)$ at the start of the feeders is set to unity.
Similarly, the voltage gradient $w(x)$ at each end of the feeders is set to zero.
The difference of voltage gradient increases toward the start of the feeders due to the boundary condition at the end of the feeders.
On the other hand, the differences of voltage amplitude and voltage phase increase toward each end of the feeders.
Note that the observation is consistent for the choice of multiple values for $\varepsilon$: $1\times10^{-1},1\times 10^{-2},\ldots,
1\times10^{-8}$.

The practical feeder model is complicated, and hence the quantitative evaluation like Table~\ref{tab:sim} is not straightforward.
Here, we use the two norms of differences in voltage amplitude between the asymptotic expansions and the direct numerical solution, which are similar to the standard $\mathcal{L}^2$ and $\mathcal{L^{\infty}}$ norms of functions.
Although mathematically not rigor, the two norms for the difference $e(x)$ are described as
\begin{equation}
|e(x)|_2 := \sqrt{\int_{\textrm{all feeders}}\{e(x)\}^2{\rm d}x}, \qquad
|e(x)|_\infty :=\max_{x\in\textrm{all feeders}} |e(x)|.
\end{equation}
The $\mathcal{L}^2$-like norm $|e(x)|_2$ implies the RMS quantification of the difference, and the $\mathcal{L}^\infty$-like norm $|e(x)|_\infty$ does the worst-case quantification.
The computational results on the two norms are shown in Table~\ref{tab:kanden}.
It is clearly shown in the table that the $\mathcal{L}^2$- and $\mathcal{L^{\infty}}$-like norms become small as the order increases.
Consequently, the proposed asymptotic expansions are capable of evaluating the practical feeder model in Figure~\ref{fig:kanden}.

\begin{table}[t]
   \centering
   \caption{
   $\mathcal{L}^2$- and $\mathcal{L^{\infty}}$-like norms
   for voltage amplitude differences for
   practical feeder model in Figure~\ref{fig:kanden}
   }
   \label{tab:kanden}
   \begin{tabular}{lcc}\hline
        & $\mathcal{L}^2$-like norm & $\mathcal{L^{\infty}}$-like norm \\\hline
       asymptotic expansion up to 1st-order & 1.9521 & 0.0154 \\\hline
       asymptotic expansion up to 2nd-order & 0.4636 & 0.0037 \\\hline
       asymptotic expansion up to 3rd-order & 0.2471 & 0.0019 \\\hline
       asymptotic expansion up to 4th-order & 0.2113 & 0.0017 \\\hline
   \end{tabular}
\end{table}

\begin{figure}[h]
   \centering
   \includegraphics[width=0.85\textwidth]{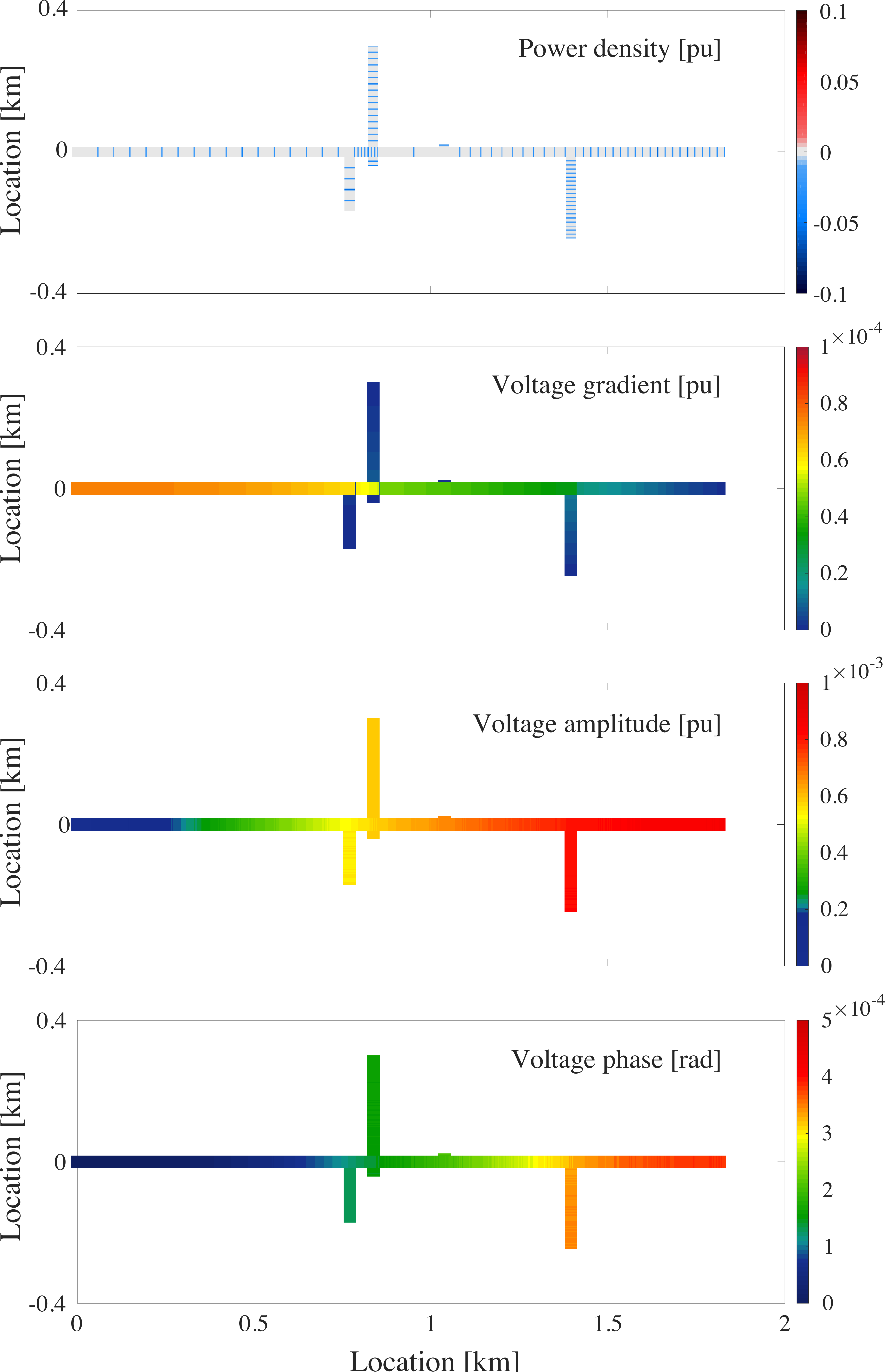}
   \caption{
   Visualization of the differences between asymptotic expansion \eqref{eqn:per} up to 4th and the nonlinear ODE \eqref{eqn:ode}.
   The power density function $p(x)$ used for simulations is also visualized on the top.
   }
   \label{fig:sim_kanden}
\end{figure}

\section{Concluding Remarks}
\label{sec:outro}

Motivated by recent electrification of vehicle and its impact to power distribution grids, in this paper we revisited the TPBV formulation of nonlinear ODE for the assessment problem of distribution voltage profiles.
Its asymptotic assessment was newly proposed by applying the regular perturbation technique in ODEs, and its effectiveness was established with numerical simulations of the simple and practical configurations of the power distribution grid.
The key derivation is the asymptotic expansion \eqref{eqn:per} of solutions of the nonlinear TPBV problem and provides an analytical insight that can explain the impact of EV charging on the distribution voltage profile.

Several remarks on the work in this paper are presented.
First, it is desirable to prove any convergence theorem for the asymptotic expansion.
Our numerics suggest that the accuracy of asymptotic expansion could be improved by including higher-order perturbation terms.
Second, it is of scientific interest and technological significance to consider the impact of the time-dependent variation of EV charging on the distribution voltage profile.
Third, it is also interesting to connect the asymptotic assessment with the regulation of distribution voltage profile against such variation.

\section*{Acknowledgement}
The authors would like to thank Mr. Shota Yumiki (Osaka Prefecture University) for valuable discussions on the work presented in this paper.

\appendix
\def\thesection{Appendix \Alph{section}}
\section{Proof of Theorem~\ref{thm:TPBV}}
\label{sec:AppA}

In this proof, we consider the original nonlinear ODE \eqref{eqn:ode} with \eqref{eqn:hoge}, where $p(x),q(x)\in C^0[0,L]$, and $\theta,s,w\in\mathbb{R}$ and $v\in\mathbb{R}_{>0}$ (set of all positive real numbers).
The current proof is devoted to the existence of solutions of the nonlinear TPBV problem described by \eqref{eqn:ode} and \eqref{eqn:ode_boundary}.

First of all, we consider the differential equation \eqref{eqn:ds/dx} for determining $s(x)$, which is an IV problem.
Because of $p(x),q(x)\in C^0[0,L]$, the right-hand side of \eqref{eqn:ds/dx} can be explicitly integrated in $x$, and $s(x)$ is expressed in a self-consistent manner as follows:
\[
s(x) = s(L)+\int_{L}^{x}\frac{Bp(\xi)-Gq(\xi)}{Y^2}\dd\xi \qquad \forall x\in[0,L].
\]
Thus, $s(x)$ is unique and $C^1$.

Next, in order to determine $v(x)$ and $w(x)$, it is necessary consider the nonlinear TPBV problem as
\begin{equation}
\left.
\begin{aligned}
\frac{\dd v}{\dd x} &= w\\
\frac{\dd w}{\dd x} &= \frac{s(x)^2}{v^3}-\varepsilon\frac{G\tilde{p}(x)+B\tilde{q}(x)}{v(G^2+B^2)}
\end{aligned}
\right\}, \quad \forall x\in[0,L],
\label{eqn:ode_vw}
\end{equation}
with
\begin{equation}
v(0)=1, \quad w(L)=0.
\end{equation}
Since $\tilde{p}(x)$, $\tilde{q}(x)\in C^0[0,L]$ and $s(x)\in C^1[0,L]$, the right-hand sides of (\ref{eqn:ode_vw}) have continuous first derivatives with respect to $v$ and $w$ for $(x,v,w)\in {\cal D}$ (an open set in $[0,L]\times(\mathbb{R}_{>0}\times\mathbb{R})$).
In the TPBV problem, we fix $\varepsilon$ as non-negative.

To consider the existence of solutions of the nonlinear TPBV problem \eqref{eqn:ode_vw}, let us define the initial-value problem for (\ref{eqn:ode_vw}) as
\begin{equation}
\left.
\begin{aligned}
\frac{\dd v}{\dd x} &= w\\
\frac{\dd w}{\dd x} &= \frac{s(x)^2}{v^3}-\lambda\frac{G\tilde{p}(x)+B\tilde{q}(x)}{v(G^2+B^2)}
\end{aligned}
\right\},
\label{eqn:ode2}
\end{equation}
with
\begin{equation}
v(0)=1, \quad w(0)=\eta,
\end{equation}
where $\eta$ is picked up from $\cal D$.
The parameter $\lambda$ is picked up from an open interval in $\mathbb{R}$ including $0$.
According to the standard theorems of existence and uniqueness of solutions for IV problems \cite{Hale}, there exists an unique solution $(v(x, 0, (1,\eta), \lambda),w(x, 0, (1,\eta), \lambda))$ of (\ref{eqn:ode2}) passing through $(0,(1,\eta))$.
This solution can be extend to $x=L$.
Also, from the standard theorem on the dependence of solutions on parameters and initial data \cite{Hale},  the solution $w(x, 0, (1,\eta), \lambda))$ is continuously differentiable with respect to $\eta$ and $\lambda$ in its domain of definition.
At $\lambda=0$ (as $\varepsilon\to +0$), because of $s(x)=0,~\forall x\in[0,L]$ (see the ODE \eqref{eqn:ds/dx} for $s$), the solution $w(x, 0, (1,\eta), 0)$ is exactly $\eta$.

Now, we are in a position to prove the existence of solutions of the nonlinear TPBV problem \eqref{eqn:ode_vw}.
For this, we define
\begin{equation}
\phi(\eta,\lambda):=w(L, 0, (1,\eta), \lambda).
\end{equation}
The proof is that we find a solution of
\begin{equation}
\phi(\eta,\lambda)=0,\qquad \eta\in{\cal D},~\lambda>0.
\label{eqn:goal}
\end{equation}
For this, we use the implicit function theorem \cite{Hale}.
First, we see $\phi(\eta,0)=\eta=0$.
From above, the derivatives $\DD\phi/\DD\eta$ and $\DD\phi/\DD\lambda$ are continuous in an open set including $(\eta,\lambda)=(0,0)$.
In addition, the value of $\DD\phi/\DD\eta$ estimated at $(\eta,\lambda)=(0,0)$ is exactly one (not zero).
Therefore, from the implicit function theorem, there exists a map $\eta^\ast$ from an open interval including $0$ to $\mathbb{R}$  such that $\phi(\eta^\ast(\lambda),\lambda)=0$, $\eta^\ast(0)=0$, and $(\dd\eta^\ast/\dd\lambda)_{\lambda=0}\neq 0$.
Because of the open interval, this implies that there exists a solution of (\ref{eqn:goal}) for a positive $\lambda$, i.e. $\varepsilon$.
This proves that there exists $\varepsilon>0$ such that the nonlinear TPBV problem \eqref{eqn:ode_vw} has a solution.
The solution, simply denoted by $v(x),w(x)$, is $C^1$ from the above argument of IV problems on ODEs.

Finally, since $s(x),v(x)$ are $C^1$ and $v(x)>0$, from the IV problem described by \eqref{eqn:dt/dx} and $\theta(0)=0$, $\theta(x)$ exists for $x\in[0,L]$ uniquely and is $C^1$.
By collecting all the statements above, it follows that there exists $\varepsilon>0$ such thhat the nonlinear TPBV problem described by \eqref{eqn:ode} and \eqref{eqn:ode_boundary} has a solution of $C^1$ regularity.

\section{Derivation of the Boundary Conditions~\eqref{eqn:perbound}}
\label{sec:AppB}

\begin{figure}
  \centering
  \includegraphics[width=0.4\hsize]{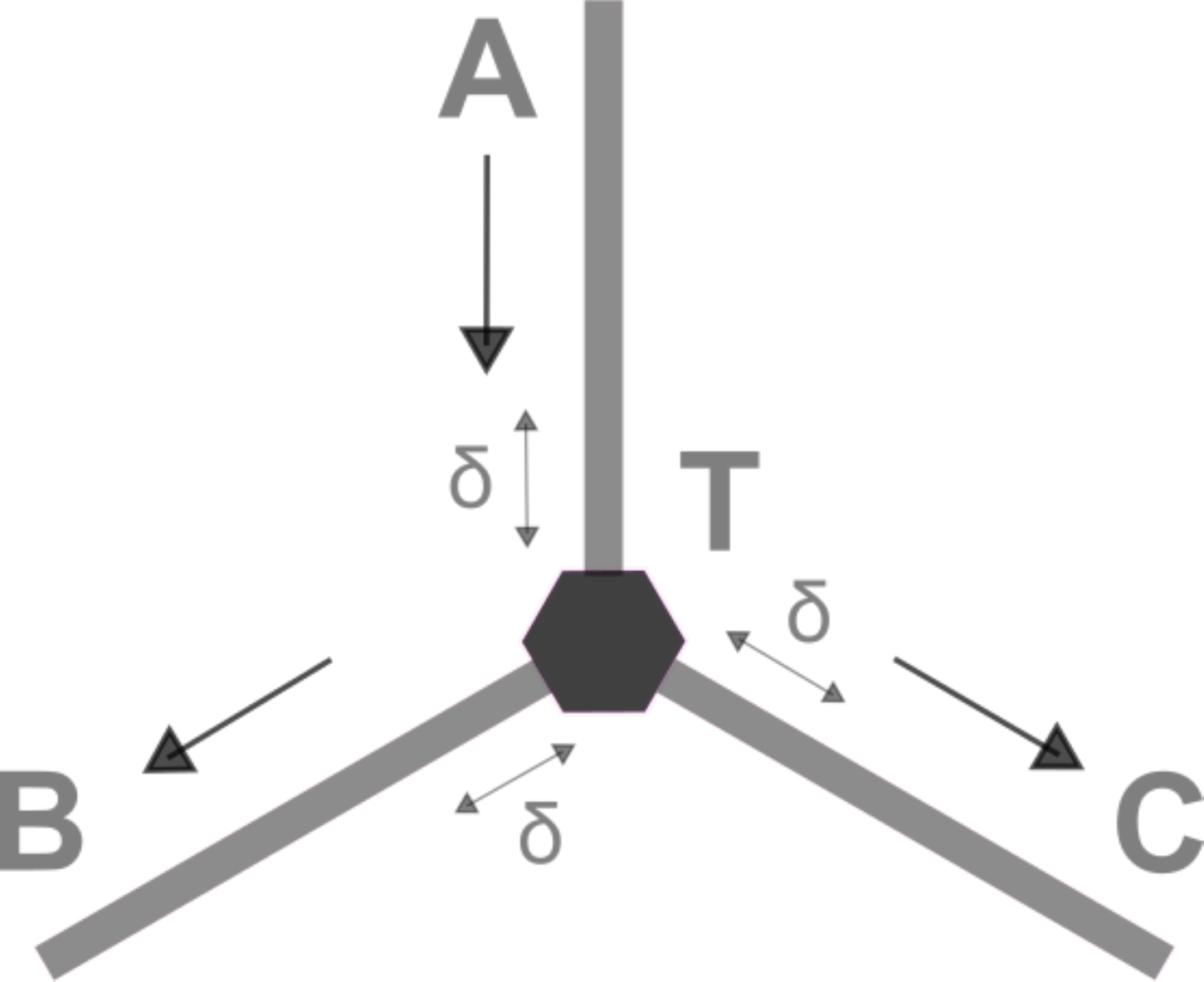}
  \caption{
  Configuration of three feeders A, B, and C connected at one bifurcation point T.
  The arrows represent the reference directions of current flows.
  We consider the neighborhood at the short distance $\delta$ from T.
  }
  \label{fig:branch_app}
\end{figure}

This appendix is devoted to the derivation of the boundary conditions \eqref{eqn:perbound} for the first order.
The derivation for higher-order cases is the same as below.
We refer to the bifurcation point in Figure \ref{fig:branch_app} as T.
We also represent the value of $\theta_{j,1}$ at the point separated from T by $\delta\in\mathbb{R}_{>0}$ along the feeder $j\in \{\rm A,B,C\}$, as $\theta_{j,1}(\delta)$ where $\delta=0$ implies $x_j=0$.
The same notation is used for the other dependent variables.
From the continuity of $\theta_{j,1}$ and $v_{j,1}$ at T, the values $\theta\sub{T}$ and $v\sub{T}$ at T are defined as follows:
\begin{align}
  \theta\sub{T,1}
  &:= \lim_{\delta\rightarrow+0}\theta\sub{A,1}(\delta) = \lim_{\delta\rightarrow+0}\theta\sub{B,1}(\delta) = \lim_{\delta\rightarrow+0}\theta\sub{C,1}(\delta), \\
  v\sub{T,1}
  &:= \lim_{\delta\rightarrow+0}v\sub{A,1}(\delta) = \lim_{\delta\rightarrow+0}v\sub{B,1}(\delta) = \lim_{\delta\rightarrow+0}v\sub{C,1}(\delta),
\end{align}
where we assume $v\sub{T,1}>0$ in a practical viewpoint.
Then, from \eqref{eqn:per1}, the values of the remaining variables $s_{j,1}$ and $w_{j,1}$
at T are defined as follows:
\begin{equation}
  \left.
  \begin{array}{ccl}
    \displaystyle s\sub{A,1}(0)
    &:=& \displaystyle -\lim_{\delta\rightarrow+0}
    \frac{\theta\sub{T,1}-\theta\sub{A,1}(\delta)}{\delta}
    \\
    \displaystyle s\sub{B,1}(0)
    &:=& \displaystyle -\lim_{\delta\rightarrow+0}
    \frac{\theta\sub{B,1}(\delta)-\theta\sub{T,1}}{\delta}
    \\
    \displaystyle s\sub{C,1}(0)
    &:=& \displaystyle -\lim_{\delta\rightarrow+0}
    \frac{\theta\sub{C,1}(\delta)-\theta\sub{T,1}}{\delta}
  \end{array}
  \right\},
  \label{eqn:s_abc}
\end{equation}
and
\begin{equation}
  \left.
  \begin{array}{ccl}
    \displaystyle w\sub{A,1}(0) &:=& \displaystyle \lim_{\delta\rightarrow+0}\frac{v\sub{T,1}-v\sub{A,1}(\delta)}{\delta}
    \\
    \displaystyle w\sub{B,1}(0) &:=& \displaystyle \lim_{\delta\rightarrow+0}\frac{v\sub{B,1}(\delta)-v\sub{T,1}}{\delta}
    \\
    \displaystyle w\sub{C,1}(0) &:=& \displaystyle \lim_{\delta\rightarrow+0}\frac{v\sub{C,1}(\delta)-v\sub{T,1}}{\delta}
  \end{array}
  \right\}.
  \label{eqn:w_abc}
\end{equation}
The voltage phasors at $\delta$ are also introduced with the independent variables $\theta$ and $v$: $\dot{V}_{j,1}(\delta)=v_{j,1}(\delta)\ee^{\ii\theta_{j,1}(\delta)}$ for $j\in\{\rm A,B,C\}$.
Then, the following equations are obtained from the first law of Kirchhoff.
\begin{align}
    \frac{\dot{V}_{\rm A,1}-\dot{V}_{\rm T,1}}{\delta\dot{Z}}
    -\frac{\dot{V}_{\rm T,1}-\dot{V}_{\rm B,1}}{\delta\dot{Z}}
    -\frac{\dot{V}_{\rm T,1}-\dot{V}_{\rm C,1}}{\delta\dot{Z}} &=0, \nonumber\\
    \displaystyle \frac{v_{\rm A,1}(\delta)\mathrm{e}^{\ii\{\theta_{\rm A,1}(\delta)-\theta_{\rm T,1}\}}-v_{\rm T,1}}{\delta}
    +\frac{v_{\rm B,1}(\delta)\mathrm{e}^{\ii\{\theta_{\rm B,1}(\delta)-\theta_{\rm T,1}\}}-v_{\rm T,1}}{\delta}
    & \nonumber\\
    +\frac{v_{\rm C,1}(\delta)\mathrm{e}^{\ii\{\theta_{\rm C,1}(\delta)-\theta_{\rm T,1}\}}-v_{\rm T,1}}{\delta} &=0,
  \label{eqn:Kirchhoff}
\end{align}
where $\dot{Z}$ is the impedance per unit length.
According tothe small amount of $\delta$, the trigonometric functions for $\theta_{\rm A,1}(\delta)$ and $\theta_{\rm T,1}$, it
can be written as follows:
\begin{align}
  \cos (\theta_{\rm A,1}-\theta_{\rm T,1}) &\simeq
    \displaystyle 1,
    \label{eqn:cos}
    \\
    \sin (\theta_{\rm A,1}-\theta_{\rm T,1}) &\simeq
    \displaystyle
    \theta_{\rm A,1}(\delta)-\theta_{\rm T,1}.
    \label{eqn:sin}
\end{align}
From \eqref{eqn:cos} and \eqref{eqn:sin}, we expand the first term of the left-hand side of \eqref{eqn:Kirchhoff} as follows:
\begin{align}
  \frac{v_{\rm A,1}(\delta)\mathrm{e}^{\ii\{\theta_{\rm A,1}(\delta)-\theta_{\rm T,1}\}}-v_{\rm T,1}}{\delta}
  & =
  \frac{v_{\rm A,1}(\delta)}{\delta}\biggl\{\cos(\theta_{\rm A,1}(\delta)-\theta_{\rm T,1})+\ii\sin(\theta_{\rm A,1}(\delta)-\theta_{\rm T,1})\biggr\}-\frac{v_{\rm T,1}}{\delta} \nonumber
  \\
  & \simeq
  \frac{v_{\rm A,1}(\delta)-v_{\rm T,1}}{\delta}+\ii v_{\rm A,1}(\delta)\frac{\theta_{\rm A,1}(\delta)-\theta_{\rm T,1}}{\delta}.
\end{align}
By taking the limitation $\delta\rightarrow0$, the above equation can be rewritten from \eqref{eqn:s_abc} and \eqref{eqn:w_abc} to the following:
\begin{align}
  \lim_{\delta\rightarrow+0}
  \displaystyle \frac{v_{\rm A,1}(\delta)\mathrm{e}^{\ii\{\theta_{\rm A,1}(\delta)-\theta_{\rm T,1}\}}-v_{\rm T,1}}{\delta}
  &= -w_{\rm A,1}(0)+\ii v_{\rm T,1}s_{\rm A,1}(0).
\end{align}
The same derivation holds for the feeders B and C, and from \eqref{eqn:Kirchhoff} we have
\begin{align}
  -w_{\rm A,1}(0)+\ii v_{\rm T,1}s_{\rm A,1}(0)+w_{\rm B,1}(0)-\ii v_{\rm T,1}s_{\rm B,1}(0)+w_{\rm C,1}(0)-\ii v_{\rm T,1}s_{\rm C,1}(0)&=0,
  \nonumber\\
  -w_{\rm A,1}(0)+w_{\rm B,1}(0)+w_{\rm C,1}(0)+\ii v_{\rm T,1}\{s_{\rm A,1}(0)-s_{\rm B,1}(0)-s_{\rm C,1}(0)\}&=0.
  \label{eqn:ws}
\end{align}
This clearly show \eqref{eqn:perbound}, namely,
\begin{equation}
  w_{\rm A,1}-w_{\rm B,1}-w_{\rm C,1}=0, \quad
  s_{\rm A,1}-s_{\rm B,1}-s_{\rm C,1}=0.
\end{equation}

\section{Representation of Voltage Regulation Devices}
\label{sec:AppC}

In this section, we consider the new boundary conditions for the nonlinear TPBV problem in the case where a voltage regulation device (Step Voltage Regulator \cite{Kersting}; SVR) is installed in the feeder.
In this case, the voltage amplitude changes discontinuously at the installed location of SVR.
This poses a new boundary condition at the location.
In Figure \ref{fig:simplefeeder}, we consider that a single-phase tap-changing transformer with variable turn ration $n$ ($>0$) is installed at the location $x=\xi\in(0,L)$.
Here, for the simple derivation, we assume no power supply and demand at the location $x=\xi$, that is, $p(\xi)=q(\xi)=0$.
Thus, the following boundary conditions are derived in \cite{Susuki2}:
\begin{equation}
   \left.
   \begin{array}{ccc}
      \displaystyle \theta(\xi-) &=& \theta(\xi+) 
      \\
       v(\xi-) &=& \displaystyle\frac{1}{n}v(\xi+)
      \vspace{1mm}
      \\\displaystyle s(\xi-) &=& s(\xi+) 
      \\
       w(\xi-) &=& n\cdot w(\xi+)
   \end{array}
   \right\}
   \label{eqn:svr}
\end{equation}
where the limit values for the location of the transformer such as $\theta(\xi-)$ and $v(\xi+)$ are represented by $\theta(\xi-)=\displaystyle \lim_{\delta\to-0}\theta(\xi+\delta)$ or $v(\xi+)=\displaystyle \lim_{\delta\to+0}v(\xi+\delta)$.
By a similar manner as above, the condition \eqref{eqn:svr} at the location of SVR is re-written in the asymptotic framework as follows:
\begin{equation}
   \left.
   \begin{array}{ccccccccl}
      \displaystyle \theta_1(\xi-) &=& \theta_1(\xi+), && \theta_2(\xi-) &=& \theta_2(\xi+), && \cdots \qquad
      \vspace{1mm}
      \\\displaystyle v_1(\xi-) &=& \displaystyle\frac{1}{n}v_1(\xi+), && v_2(\xi-) &=& \displaystyle \frac{1}{n}v_2(\xi+), && \cdots \qquad
      \vspace{1mm}
      \\\displaystyle s_1(\xi-) &=& s_1(\xi+), && s_2(\xi-) &=& s_2(\xi+), && \cdots \qquad
      \vspace{1mm}
      \\\displaystyle w_1(\xi-) &=& n\cdot w_1(\xi+), && w_2(\xi-) &=& n\cdot w_2(\xi+), && \cdots \qquad
   \end{array}
   \right\}.
   \label{eqn:svr_per}
\end{equation}
The above derivation for the perturbation terms is almost the same as in \eqref{eqn:perbound} and is hence omitted in this paper.

\section{Detailed Settings of Numerical Demonstrations}
\label{sec:AppD}

The appendix is described the detailed settings of the numerical demonstrations in Section~\ref{sec:numerics} based on Mizuta \emph{et al.} \cite{Mizuta2}.

In both the feeder models, we assume that the secondary voltage of the transformer is set as 6.6\,kV, which is the normal condition in Japan's high-voltage distribution networks.
The condactance and susceptance of each feeder are common and constant in $x$ as in \eqref{eqn:hoge}.
The feeder's resistance (or reactance) are set at 0.227\,$\Omega$/km (or 0.401\,$\rm\Omega$/km).
These values are from \cite{Mizuta2} and based on the standard setting in Japan.
In the following, we use per-unit system \cite{Kersting} for numerical simulations of the nonlinear ODE model.
The conductance $G$ (or susceptance $B$) of the simple model per unit-length is calculated as 3.881 (or 6.856) in per-unit system ($G/B$ is about 5.661$\times10^{-1}$; it is typical in practice).
Similarly, $G$ (or $B$) of the practical model is also calculated as 2.329 (or 4.113).

For the practical feeder model in Figure~\ref{fig:kanden},
it has multiple bifurcation points and 9 charging stations denoted by \emph{circled numbers}.
The sum of the lengths of all the feeders is 2.52\,km.
The secondary voltage at the bank is regulated at 6.6\,kV, and the loading capacity of the bank is set at 20\,MVA.
The model has 103 pole transformers distributed along the feeders.
All the pole transformers are connected to a total of 1001 residential households.
The residential loads used here are based on practical measurement at 19 o'clock in summer in Japan.
The detailed data on the feeders and residential loads could not be published following an agreement with the utility company.

\input{bib.tex}

\end{document}